\documentclass[11pt]{article}
\usepackage{amsmath, amssymb, amsfonts, amsthm, verbatim, dsfont}

\newtheorem{theorem}{Theorem}[section]
\newtheorem{proposition}[theorem]{Proposition}

\newtheorem{lemma}[theorem]{Lemma}

\newtheorem{remark}[theorem]{Remark}

\newcommand{\ud}{\mathrm{d}}

\newcommand{\tr}{\mathrm{tr}}

\newcommand{\vp}{\varphi}
\newcommand{\vpt}{\varphi_{t}}
\newcommand{\vpta}{\varphi_{t}^{(\alpha)}}
\newcommand{\vps}{\varphi_{s}}
\newcommand{\vpsa}{\varphi_{s}^{(\alpha)}}
\newcommand{\vptan}{\varphi_{t}^{(\alpha_N)}}
\newcommand{\ovptan}{\overline{\varphi}_t^{(\alpha_N)}}

\newcommand{\bR}{{\mathbb R}}
\newcommand{\bN}{{\mathbb N}}
\newcommand{\bC}{{\mathbb C}}

\newcommand{\hoart}{H^{1}_{A}(\mathbb{R}^{3})}
\newcommand{\hoa}{H^{1}_{A}}

\newcommand{\pt}{\phi_t}
\newcommand{\opt}{\overline{\phi}_t}

\newcommand{\pn}{\psi_{N}}
\newcommand{\pnt}{\psi_{N, t}}

\newcommand{\gnt}{\gamma_{N, t}}
\newcommand{\gntk}{\gamma_{N, t}^{(k)}}

\newcommand{\gnto}{\gamma_{N, t}^{(1)}}

\newcommand{\xk}{\mathbf{x}_{k}}
\newcommand{\xkd}{\mathbf{x}_{k}'}
\newcommand{\xNk}{\mathbf{x}_{N-k}}

\newcommand{\cha}{{\cal{H}}_{A}}
\newcommand{\chna}{{\cal{H}}_{N}^{\alpha}}
\newcommand{\cF}{{\cal{F}}}
\newcommand{\cN}{{\cal{N}}}
\newcommand{\cU}{{\cal{U}}}
\newcommand{\cL}{{\cal{L}}}
\newcommand{\cK}{{\cal{K}}}
\newcommand{\cW}{{\cal{W}}}
\newcommand{\cM}{{\cal{M}}}

\newcommand{\im}{\mbox{Im}}

\begin{document}
\title{Mean-field quantum dynamics with magnetic fields}
\author{Jonas L\"uhrmann \\
\\
Department of Mathematics, ETH Zurich, \\
R\"amistrasse 101, CH-8092 Zurich, Switzerland}

\maketitle

\begin{abstract}
 We consider a system of $N$ bosons in three dimensions interacting through a mean-field Coulomb potential in an external magnetic field. For initially factorized states we show that the one-particle density matrix associated with the solution of the $N$-body Schr\"odinger equation converges to the projection onto the solution of the magnetic Hartree equation in trace norm and in energy as $N \rightarrow \infty$. Estimates on the rate of convergence are provided.
\end{abstract}

\section{Introduction}
\setcounter{equation}{0}

We investigate the mean-field quantum dynamics of a system of $N$ identical and spinless bosons in three dimensions subject to an external magnetic field. The state of the system is described by a symmetric wave function $\pn \in L^2(\bR^{3N})$ with $\|\pn\|_2 = 1$. We consider two-particle Coulomb interactions. The external magnetic field is generated by a magnetic vector potential $A: \bR^3 \rightarrow \bR^3$. The Hamiltonian of the system is then given by
\begin{equation} \label{equ:hamiltonian}
 H_{N} \, = \, \sum_{j=1}^{N}{\, (-i \nabla_{x_{j}} + A(x_{j}))^{2}} + \frac{1}{N} \sum_{i<j}^{N}{\, \frac{\lambda}{|x_{i}-x_{j}|}},
\end{equation}
where $x_j \in \bR^3$ denotes the position of the $j$-th particle and $\lambda \in \bR$ is a coupling constant. The factor $\frac{1}{N}$ in front of the interaction potential ensures that the kinetic and potential energy have the same scaling behavior in $N$ and corresponds to very weak interactions between the particles. 

\medskip

The time evolution of the system is governed by the Schr\"odinger equation 
\begin{equation} \label{equ:schroed_equ}
 i \partial_{t} \pnt \, = \, H_{N} \pnt
\end{equation}
with initial datum $\psi_{N, t=0} \, = \, \pn$, where $\pnt$ denotes the wave function of the system at time $t$. Here and henceforth, the subscript $t$ to a quantity denotes its time-dependence. We consider factorized initial states $\pn = \vp^{\otimes N}$ for some $\vp \in L^2(\bR^3)$. Under appropriate assumptions about the magnetic vector potential $A$, the Hamiltonian \eqref{equ:hamiltonian} can be self-adjointly realized on $L^2(\bR^{3N})$. By Stone's Theorem, the unique solution to \eqref{equ:schroed_equ} is then given by $\pnt = e^{-iH_N t} \pn$. 

\medskip

When applying this model to real-world physical systems we are facing numbers of particles of several powers of ten. Owing to the large number of particles it is practically impossible to obtain any qualitative information about the behavior of the system from the solution $\pnt = e^{-i H_N t} \pn$. However, in the mean-field regime that we consider it is possible to derive effective evolution equations which are on the one hand at least numerically solvable and which on the other hand give a good approximate description of the macroscopic behavior of the system. 

\medskip

Due to the weak interactions between the particles one expects that the wave function $\pnt$ also stays factorized at later times $t > 0$, i.e. $\pnt \simeq \vpt^{\otimes N}$ in a sense to be made precise. A simple heuristic argument shows that in the limit $N \rightarrow \infty$ the one-particle wave function $\vpt$ is expected to satisfy the magnetic Hartree equation 
\begin{equation} \label{equ:hartree_equ}
          i \partial_{t} \vp_{t} \, = \, (-i \nabla + A)^{2} \vp_{t} + (\frac{\lambda}{|\cdot|} \ast |\vp_{t}|^{2}) \vp_{t}
\end{equation}
with initial datum $\vp_{t=0} = \vp$. 

\medskip

We define the density matrix $\gnt$ associated with the state $\pnt$ as the orthogonal projection onto $\pnt$, i.e. 
\begin{equation*}
 \gnt \, = |\pnt \rangle \langle \pnt|.
\end{equation*}
The operator $\gnt$ is a positive trace class operator on $L^2(\bR^{3N})$ with unit trace. For every $k \in \{1, \ldots, N\}$ we also define the corresponding $k$-particle marginal density $\gntk$ through its integral kernel
\begin{equation*}
 \gntk(\xk; \xkd) \, = \, \int_{\mathbb{R}^{3(N-k)}}{\ud \mathbf{x}_{N-k} \, \pnt(\xk, \xNk) \overline{\pnt(\xkd, \xNk)}},
\end{equation*}
where $\xk = (x_{1}, \ldots, x_{k})$, $\xkd = (x_{1}', \ldots, x_{k}') \in \mathbb{R}^{3k}$ and $\xNk = (x_{k+1}, \ldots, x_{N}) \in \mathbb{R}^{3(N-k)}$ with $x_{j}, x_{j}' \in \bR^3$ for $j \in \{ 1, \ldots, N\}$. It follows that $\gntk$ is a positive trace class operator with unit trace on $L^2(\bR^{3k})$.

\medskip

It turns out that on the level of marginal densities one can show that the limiting dynamics as $N \rightarrow \infty$ of the many-body linear Schr\"odinger equation \eqref{equ:schroed_equ} is given by the solutions to the one-body nonlinear Schr\"odinger equation \eqref{equ:hartree_equ},
\begin{equation} \label{trace_norm_mf_limit}
 \lim_{N \rightarrow \infty}{ \, \tr{\, \Bigl| \gntk - |\vpt \rangle \langle \vpt|^{\otimes k} \Bigr|}} \, = \, 0
\end{equation}
for every fixed $k \in \bN$ and every fixed $t \in \bR$. \\
Furthermore, it is of interest to show that the convergence to the limiting Hartree dynamics also holds in energy, i.e.
\begin{equation} \label{equ:energy_convergence}
 \lim_{N \rightarrow \infty}{\, \tr{\, \Bigl| (-i\nabla+A)^{\otimes k} \Bigl(\gntk - |\vpt \rangle \langle \vpt|^{\otimes k}\Bigr) (-i\nabla+A)^{\otimes k} \Bigr|}} \, = \, 0
\end{equation}
for every fixed $k \in \bN$ and every fixed $t \in \bR$. 

\medskip

The study of mean-field quantum dynamics has a relatively long history. Unless stated otherwise, the following results refer to non-relativistic systems with two-particle interactions given by an interaction potential $V$ and without an external magnetic field. 

The first results establishing a relation between the many-body Schr\"odinger evolution and the nonlinear Hartree dynamics for smooth interaction potentials $V$ were obtained by Hepp in \cite{H74}. Ginibre and Velo generalized his results to singular potentials in \cite{GV}. The first proof of the convergence \eqref{trace_norm_mf_limit} for bounded potentials $V$ was given by Spohn \cite{S80}. His method is based on expanding the BBGKY hierarchy of evolution equations for marginals. Since then progress has been made mainly in two directions: First, to show the convergence \eqref{trace_norm_mf_limit} for more singular potentials and second, to obtain estimates on the rate of convergence of \eqref{trace_norm_mf_limit}. 

In \cite{EY01}, Erd\H os and Yau generalized and extended Spohn's method to the Coulomb potential $V(x) = \frac{\lambda}{|x|}$, $\lambda \in \mathbb{R}$. Partial results in this direction had been obtained before by Bardos, Golse and Mauser (see \cite{BEGMY02} and \cite{BGM00}). The method was extended by Elgart and Schlein in \cite{ElS07} to the case of semi-relativistic systems with Coulomb interactions. See also \cite{FGS07} and \cite{FKS09} for further results.

Rodnianski and Schlein \cite{RS09} proved the convergence \eqref{trace_norm_mf_limit} for Coulomb-type interactions using an idea of Hepp \cite{H74}. They obtained an estimate on the rate of convergence of the type
\begin{equation*}
 \tr{\, \left|\gntk - |\vpt \rangle \langle \vpt|^{\otimes k}\right|} \, \leq \, \frac{C(k)}{\sqrt{N}} e^{K(k) t},
\end{equation*}
where $C(k), K(k) > 0$ are $k$-dependent constants. \\
In \cite{KP}, Knowles and Pickl obtained estimates on the rate of convergence for more singular potentials for non-relativistic systems and for Coulomb interactions for semi-relativistic systems.\\
Chen, Lee and Schlein \cite{CLS} derived optimal estimates on the rate of convergence \eqref{trace_norm_mf_limit} for one-particle marginals for non-relativistic systems with Coulomb interactions 
\begin{equation*}
 \tr{\, \Bigl| \gamma_{N, t}^{(1)} - |\vpt \rangle \langle \vpt| \Bigr| } \, \leq \, \frac{C e^{Kt}}{N},
\end{equation*}
where $C, K > 0$ are constants. \\
Michelangeli and Schlein \cite{MS} obtained the first result that the convergence to the limiting Hartree dynamics also holds in energy. For semi-relativistic systems with Coulomb interactions they proved the corresponding convergence \eqref{equ:energy_convergence} for one-particle marginals together with an estimate on the rate of convergence. 

A similar analysis has been carried out for systems with two-particle interactions that have a singular scaling in $N$ and tend to a delta-interaction as $N \rightarrow \infty$. The many-body quantum dynamics is then approximated by the Gross-Pitaevskii equation (see \cite{ESY10} and references therein). 

\medskip

In this work we extend results from \cite{KP} and \cite{MS} to the case of an external magnetic field. Throughout we will make the following assumption regarding the magnetic vector potential $A: \bR^3 \rightarrow \bR^3$ and the generated magnetic field $B = \nabla \times A$: 

\medskip

\noindent \textbf{Assumption (A).} Let $A \in C^{\infty}(\mathbb{R}^{3}; \mathbb{R}^{3})$ and define $B = \nabla \times A$. Assume that there exists $\varepsilon > 0$ such that
\begin{align*}
|\partial^{\alpha} B(x)| \, & \leq \, C_{\alpha} (1+|x|)^{-(1+\varepsilon)} && \forall \, |\alpha| \geq 1, \forall \, x \in \mathbb{R}^{3}, \\
|\partial^{\alpha} A(x)| \, & \leq \, C_{\alpha} && \forall \, |\alpha| \geq 1, \forall \, x \in \mathbb{R}^{3},
\end{align*}
where $C_{\alpha}$ are constants depending only on the multi-index $\alpha$. 

\medskip

Note that the vector potential $A(x) = \frac{1}{2} B_0 \times x$ generating a constant magnetic field $B_0$ fulfills this assumption. Also, smooth compactly supported perturbations of linear magnetic vector potentials satisfy the hypothesis. 

\medskip

In order to state our main results we need to introduce some notation. Denote $D_j \equiv (-i\partial_j + A_j)$ for $j \in \{ 1, 2, 3\}$. We define the $k$-th order magnetic Sobolev space $H_A^k(\bR^3)$ for $k \in \bN$ by
\begin{equation*}
 H^k_A(\bR^3) \, := \, \Bigl\{ \vp \in L^2(\bR^3) \, \Big| \, \| D_1^{\alpha_1} D_2^{\alpha_2} D_3^{\alpha_3} \vp \|_2 < \infty \, \, \mbox{for all} \, \alpha \in \bN_0^3 \, \mbox{with} \, |\alpha| = \sum_{j=1}^{3} \alpha_j \leq k \Bigr\}
\end{equation*}
with the norm
\begin{equation*}
 \|\vp\|_{H^k_A}^2 \, := \, \sum_{|\alpha| \leq k}{\, \| D_1^{\alpha_1} D_2^{\alpha_2} D_3^{\alpha_3} \vp \|_2^2}.
\end{equation*}

\begin{theorem} \label{thm:trace_convergence}
Let A $\in C^{\infty}(\bR^{3}; \bR^{3})$ satisfy assumption (A) and let $\vp \in H^{1}_{A}(\bR^{3})$ with $\|\vp\|_{2} = 1$. Set $\pn = \vp^{\otimes N}$. Let $\lambda \in \bR$ and let $\pnt = e^{-iH_{N}t} \pn$ be the evolution of the initial wave function $\pn$ with respect to the Hamiltonian \eqref{equ:hamiltonian}. Denote by $\gntk$ the $k$-particle marginals associated with $\pnt$ and denote by $\vpt$ the solution to the initial value problem for the magnetic Hartree equation \eqref{equ:hartree_equ} with initial datum $\vp_{t=0} = \vp$. Then there exists a constant $C > 0$ such that, for $k \in \bN$ and $t \in \bR$,
\begin{equation}
 \tr{ \, \Bigl| \gntk - |\vpt \rangle \langle \vpt|^{\otimes k} \Bigr| } \, \, \leq \, \, \sqrt{8} \, \sqrt{\frac{k}{N}} \, e^{C t}
\end{equation}
holds for all $N \geq k$. In particular, this implies for every fixed $k \in \mathbb{N}$ and every fixed $t \in \mathbb{R}$
\begin{equation}
 \lim_{N \rightarrow \infty}{ \tr{ \, \Bigl| \gntk - |\vpt \rangle \langle \vpt|^{\otimes k} \Bigr| } } \, = \, 0.
\end{equation}
\end{theorem}

\medskip

Moreover, we show that on the level of the one-particle marginals the convergence of the many-body linear dynamics to the Hartree dynamics also holds in energy as $N \rightarrow \infty$. Due to technical reasons we have to introduce a regularization of the Coulomb interaction potential that vanishes in the limit $N \rightarrow \infty$. For a sequence $\alpha = (\alpha_N)_{N \in \bN}$ with $\alpha_N > 0$ for all $N \in \bN$ and $\alpha_N \rightarrow 0$ as $N \rightarrow \infty$, we define the regularized Hamiltonian
\begin{equation} \label{equ:regularised_hamiltonian}
 H_N^{\alpha} \, = \, \sum_{j=1}^N (-i\nabla_{x_j} + A(x_j))^2 + \frac{1}{N} \sum_{i<j}^N \frac{\lambda}{|x_i - x_j| + \alpha_N}.
\end{equation}

\begin{theorem} \label{thm:energy_convergence}
Let $A \in C^{\infty}(\bR^{3}; \bR^{3})$ satisfy assumption (A) and let $\vp \in H_A^3(\bR^{3})$ with $\|\vp\|_{2} = 1$. Set $\pn = \vp^{\otimes N}$. Consider an arbitrary sequence $(\alpha_N)_{N \in \bN}$ with $\alpha_N > 0$ for all $N \in \bN$ and such that $N^{\beta} \alpha_N \rightarrow \infty$ as $N \rightarrow \infty$ for some $\beta > 0$. Let $\lambda \in \bR$ and let $\pnt = e^{-i H_N^{\alpha} t} \pn$ be the evolution of the initial wave function $\pn$ with respect to the regularized Hamiltonian \eqref{equ:regularised_hamiltonian}. Let $\gamma_{N, t}^{(1)}$ be the one-particle marginal associated with $\pnt$. \\
Denote by $\vpt$ the solution to the initial value problem for the magnetic Hartree equation \eqref{equ:hartree_equ} with initial datum $\vp_{t=0} = \vp$. Fix $T > 0$. Then there exists a constant $C \equiv C(T, \|\vp\|_{H^3_A})$ such that
\begin{equation}
 \tr{ \, \Bigl| (-i\nabla + A) \bigl( \gamma_{N, t}^{(1)} - |\vpt\rangle\langle\vpt| \bigr) (-i\nabla + A) \Bigr| } \, \leq \, C \Bigl( \frac{1}{N^{1/4}} + \alpha_N^{1/4} \Bigr)
\end{equation}
for all $t \in \bR$ with $|t| \leq T$, and for all $N$ sufficiently large. In particular, it follows for fixed $t \in \bR$ that $\gamma_{N, t}^{(1)} \rightarrow |\vpt\rangle\langle\vpt|$ in energy norm as $N \rightarrow \infty$.
\end{theorem}

\medskip

This paper is organized as follows. In Section \ref{sec:magnetic_hartree_equation} we show global well-posedness of the magnetic Hartree equation \eqref{equ:hartree_equ} for all magnetic vector potentials satisfying assumption (A). To this end we use magnetic Strichartz estimates by Yajima \cite{Yaj} that require the assumption (A). Furthermore, we prove several properties of the solutions to \eqref{equ:hartree_equ} that will be needed in the proof of Theorem \ref{thm:energy_convergence}. In Section \ref{sec:trace_convergence} we prove Theorem \ref{thm:trace_convergence} using a result from \cite{KP}. In Section \ref{sec:energy_convergence} we derive Theorem \ref{thm:energy_convergence} adapting the method in \cite{MS} to the magnetic case. 

\medskip

{\it Acknowledgments.} The author is indebted to L. Erd\H os and A. Michelangeli for many helpful discussions. He would like  to thank M. Struwe for stimulating discussions related to the magnetic Hartree equation.

\section{The magnetic Hartree equation} \label{sec:magnetic_hartree_equation}
\setcounter{equation}{0}

The well-posedness in $H^1_A$ of the magnetic Hartree equation with the nonlinearity $(V \ast |\vpt|^2)\vpt$, where $V \in L^p + L^{\infty}$, $p \geq 1$, was studied by Cazenave and Esteban in \cite{Caz} for an explicit linear magnetic vector potential. The proof relies on the fact that magnetic Strichartz estimates for the propagator $e^{-it(-i\nabla+A)^2}$ can be derived from an explicit formula for the propagator in the case of linear magnetic vector potentials. In this section we extend their results to the class of magnetic vector potentials satisfying assumption (A). To this end we employ short-time magnetic Strichartz estimates by Yajima \cite{Yaj}.

\begin{proposition} \label{prop:gwp_hartree}
Let $V \in L^{3/2}(\bR^{3}) + L^{\infty}(\bR^{3})$ be real-valued and even. Let A $\in C^{\infty}(\bR^{3}; \bR^{3})$ satisfy assumption (A). Choose $\vp \in H^{1}_{A}(\bR^{3})$. Then the initial value problem 
\begin{equation} \label{equ:initial_value_problem}
 \left\{ \begin{aligned}
          i \partial_{t} \vp_{t} \, & = \, (-i \nabla + A)^{2} \vp_{t} + (V \ast |\vp_{t}|^{2}) \vp_{t}, \\
          \vp_{t=0} \, & = \, \vp,
         \end{aligned}
\right. 
\end{equation}
is globally well-posed in $H^{1}_{A}$, i.e. it has a unique solution $\vpt \in C(\bR; \hoa) \, \cap \, C^1(\bR; H^{-1}_A)$ and the solution depends continuously on the initial data. Moreover, the mass $M(\vp) \, = \, \|\vp\|_2$ and the energy
\begin{equation*}
 E(\vp) \, = \, \frac{1}{2} \|(-i\nabla+A) \vp\|_2^2 + \frac{1}{4} \int_{\bR^3}{\, (V \ast |\vp|^2) |\vp|^2}
\end{equation*}
are conserved.
\end{proposition}

Recently, Cao \cite{Cao} showed global well-posedness in $\hoa$ for the magnetic Hartree equation \eqref{equ:hartree_equ} with a repulsive Coulomb interaction potential for all $A \in L^2_{loc}(\bR^3; \bR^3)$ such that $(-i\nabla + A)^2$ is self-adjoint on $L^2(\bR^3)$. The proof is based on establishing the local Lipschitz continuity in $\hoa$ of the corresponding Hartree nonlinearity. Compared with \cite{Cao}, Proposition \ref{prop:gwp_hartree} yields global well-posedness of the magnetic Hartree equation for less general magnetic vector potentials $A$, but for more general interaction potentials $V$.

\medskip

We will be repeatedly using the following properties of $\hoart$. Let $A \in L^{2}_{loc}(\bR^3; \bR^3)$ and let $\vp \in \hoart$. Then $|\vp| \in H^{1}(\bR^3)$ and the diamagnetic inequality
\begin{equation} \label{equ:diamagnetic_inequality}
 \left| \nabla |\vp| (x) \right| \, \leq \, \left| (-i\nabla + A)\vp(x) \right|
\end{equation}
holds pointwise for almost every $x \in \bR^3$ (see e.g. \cite{LL}). Thus, we have the embedding $\hoa(\bR^3) \hookrightarrow L^6(\bR^3)$ by the Sobolev inequality. Using \eqref{equ:diamagnetic_inequality} and the Hardy inequality, we obtain the so-called magnetic Hardy inequality for all $\vp \in \hoart$,
\begin{equation} \label{equ:magnetic_hardy_inequality}
 \frac{1}{4} \, \int_{\bR^3}{\ud x \, \frac{|\vp(x)|^2}{|x|^2}} \, \leq \, \int_{\bR^3}{\ud x \, |\nabla |\vp(x)||^2} \, \leq \, \int_{\bR^3}{\ud x \, |(-i\nabla + A)\vp(x)|^{2} }. 
\end{equation}

\medskip

In what follows, $a \lesssim b$ denotes $a \leq C b$, where $C$ is a positive constant that can depend on fixed parameters.

\begin{proof}[Proof of Proposition \ref{prop:gwp_hartree}] 
\textit{Local well-posedness:} \\
Local well-posedness of \eqref{equ:initial_value_problem} follows with standard techniques for nonlinear Schr\"odinger equations (see e.g. \cite{Caz}). The crucial ingredient to apply the methods from \cite{Caz} is to show a priori uniqueness to the initial value problem \eqref{equ:initial_value_problem}. To this end we use short-time Strichartz estimates for the propagator $e^{-it(-i\nabla+A)^2}$ that were established in \cite{Yaj} under the assumption (A) about the magnetic vector potential $A$. \\
\textit{Global well-posedness:} \\
Let $\vpt$ be the maximal solution to \eqref{equ:initial_value_problem} defined on the interval $I_{max} \ni 0$. We now show that the $\hoa$-norm of $\vpt$ is uniformly bounded on $I_{max}$. By the blow-up alternative, $\vpt$ must then exist globally in time. For $t \in I_{max}$ we have
\begin{equation} \label{equ:estimate_energy_below}
 E(\vpt) \, = \, \frac{1}{2} \|(-i\nabla + A)\vpt\|_2^2 + \frac{1}{4} \int_{\bR^3}{\, (V \ast |\vpt|^2) |\vpt|^2} \, \geq \, \frac{1}{2} \|(-i\nabla + A)\vpt\|_2^2 - \frac{1}{4} \Bigl\| V \ast |\vpt|^2 \Bigr\|_{\infty} \|\vpt\|_2^2.
\end{equation}
Let $V \equiv V_1 + V_2$ with $V_1 \in L^{3/2}(\bR^3)$ and $V_2 \in L^{\infty}(\bR^3)$. Then, by the Sobolev inequality and \eqref{equ:diamagnetic_inequality}, we obtain for $x \in \bR^3$,
\begin{equation*}
 \Bigl| \Bigl( V \ast |\vpt|^2 \Bigr)(x) \Bigr| \, \leq \, \|V_1\|_{3/2} \|\vpt\|_6^2 + \|V_2\|_{\infty} \|\vpt\|_2^2 \, \lesssim \|V_1\|_{3/2} \|(-i\nabla + A)\vpt\|_2^2 + \|V_2\|_{\infty} \|\vpt\|_2^2.
\end{equation*}
It is easy to see that we can decompose $V$ in such a way that $\|V_1\|_{3/2}$ can be chosen arbitrarily small. Thus, from \eqref{equ:estimate_energy_below} and the conservation of mass and energy we get for all $t \in I_{max}$
\begin{equation*}
 \|(-i\nabla+A) \vpt\|_2^2 \, \leq \,4 ( E(\vp) + \|V_2\|_{\infty} \|\vp\|_2^4),
\end{equation*}
which proves the claim.
\end{proof}

In the proof of Theorem \ref{thm:energy_convergence} we have to consider the regularized magnetic Hartree equation
\begin{equation} \label{equ:regularised_hartree_equ}
          i \partial_{t} \vpta \, = \, (-i \nabla + A)^{2} \vpta + (\frac{\lambda}{|\cdot| + \alpha} \ast |\vpta|^{2}) \vpta
\end{equation}
for $\alpha \geq 0$ with initial datum $\vp_{t=0}^{(\alpha)} = \vp \in \hoa(\bR^3)$. Below we derive properties of its solutions that will be needed in the proof of Theorem \ref{thm:energy_convergence}.

\begin{remark}
 It follows immediately from Proposition \ref{prop:gwp_hartree} that the magnetic Hartree equation \eqref{equ:hartree_equ} and the regularized magnetic Hartree equation \eqref{equ:regularised_hartree_equ} are globally well-posed in $\hoa$. 
\end{remark}

\begin{proposition} \label{prop:closeness_hartree_dynamics}
Choose $\vp \in \hoa(\bR^3)$ with $\|\vp\|_{2} = 1$ and let $\vpt$ denote the solution to the magnetic Hartree equation \eqref{equ:hartree_equ} with initial datum $\vp_{t=0} = \vp$. For $\alpha \geq 0$, let $\vpta$ denote the solution to the regularized magnetic Hartree equation \eqref{equ:regularised_hartree_equ} with initial datum $\vp_{t=0}^{(\alpha)} = \vp$. Let $T > 0$. Then we have:
\begin{itemize}
 \item[(i)] There exists a constant $C_{1} \equiv C_{1}(\|\vp\|_{\hoa})$ such that
  \begin{equation} \label{equ:uniform_hoa_bound}
   \|\vpta\|_{\hoa} \, \leq \, C_{1} \quad for \: all \: t \in \bR \: and \: all \: 0 \leq \alpha \leq 1.
  \end{equation}
 \item[(ii)] There exists a constant $C_{2} \equiv C_{2}(T, \|\vp\|_{\hoa})$ such that
  \begin{equation} \label{equ:l2_closeness}
   \|\vpt - \vpta\|_{2} \, \leq \, C_{2} \, \alpha \quad  for \: all \: |t| \leq T \: and \: all \: 0 \leq \alpha \leq 1.
  \end{equation}
 \item[(iii)] There exists a constant $C_{3} \equiv C_{3}(T, \|\vp\|_{\hoa})$ such that
  \begin{equation} \label{hoa_closeness}
   \| \vpt - \vpta \|_{\hoa} \, \leq \, C_{3} \, \alpha^{1/4} \quad for \: all \: |t| \leq T \: and \: all \: 0 \leq \alpha \leq 1.
  \end{equation}
\end{itemize}
\end{proposition}
\begin{proof}
We follow the argument of the proof of Proposition 2.2 in \cite{MS} and adapt it to the magnetic case.

\medskip

\noindent (i) follows from the conservation of energy both for the magnetic Hartree equation \eqref{equ:hartree_equ} and the regularized magnetic Hartree equation \eqref{equ:regularised_hartree_equ} and an inspection of the corresponding energy functionals. 

\medskip

\noindent (ii) We write $\vpt$ and $\vpta$ in their Duhamel expansions
  \begin{equation} \label{equ:duhamel_expansion_vpt}
   \vpt \, = \, e^{-i t (-i\nabla + A)^{2}} \vp - i \lambda \int_{0}^{t}{\ud s \, e^{-i (t-s) (-i\nabla + A)^{2}} } (\frac{1}{|\cdot|} \ast |\vp_{s}|^{2}) \vp_{s}
  \end{equation}
  and
  \begin{equation} \label{equ:duhamel_expansion_vpta}
   \vpta \, = \, e^{-i t (-i\nabla + A)^{2}} \vp - i \lambda \int_{0}^{t}{\ud s \, e^{-i (t-s) (-i\nabla + A)^{2}} } (\frac{1}{|\cdot| + \alpha} \ast |\vp_{s}^{(\alpha)}|^{2}) \vp_{s}^{(\alpha)}.
  \end{equation}
 Using the magnetic Hardy inequality \eqref{equ:magnetic_hardy_inequality} and the uniform $\hoa$-norm control \eqref{equ:uniform_hoa_bound} repeatedly, we then obtain 
 \begin{equation}
  \begin{split}
  \|\vpt-\vpta\|_2 \, &\leq \, |\lambda| \, \int_0^t{\ud s \, \Bigl\| \Bigl( \frac{1}{|\cdot|} \ast |\vps|^2 \Bigr) \vps - \Bigl( \frac{1}{|\cdot|+\alpha} \ast |\vpsa|^2 \Bigr) \vpsa \Bigr\|_2} \\
  &\leq \, |\lambda| \, \int_0^t \ud s \, \Bigl\{ \Bigl\| \Bigl(\frac{1}{|\cdot|} \ast |\vps|^2 \Bigr) (\vps - \vpsa)\Bigr\|_2 + \Bigl\| \Bigl(\frac{\alpha}{|\cdot| (|\cdot|+\alpha)} \ast |\vps|^2\Bigr) \vpsa \Bigr\|_2  \\
  & \quad \quad \quad \quad \quad \quad \quad \quad + \Bigl\| \Bigl( \frac{1}{|\cdot| + \alpha} \ast (|\vps|^2 - |\vpsa|^2) \Bigr) \vpsa \Bigr\|_2 \Bigr\} \\
  &\lesssim \, |\lambda| \, \int_0^t \ud s \, \Bigl\{ \|\vps\|_{\hoa}^2 \|\vps - \vpsa\|_2 + \alpha \, \|\vps\|_{\hoa}^2 \|\vpsa\|_2 \\
  & \quad \quad \quad \quad \quad \quad \quad \quad + (\|\vps\|_{\hoa} + \|\vpsa\|_{\hoa}) \|\vps - \vpsa\|_2 \|\vpsa\|_2 \Bigr\} \\
  &\lesssim \, |\lambda| \, \int_0^t \ud s \, \Bigl\{ \alpha + \|\vps - \vpsa\|_2 \Bigr\}.
  \end{split}
 \end{equation}
By Gronwall's lemma we find $C_2 \equiv C_2(T, \|\vp\|_{\hoa})$ such that
\begin{equation}
 \|\vpt-\vpta\|_2 \, \leq C \, \alpha \quad  for \: all \: |t| \leq T \: and \: all \: 0 \leq \alpha \leq 1.
\end{equation}

\noindent (iii) It is enough to show that there exists a constant $C \equiv C(T, \|\vp\|_{\hoa})$ such that
  \begin{equation}
    \| (-i\nabla + A) (\vpt - \vpta) \|_{2} \, \leq \, C \, \alpha^{1/4} \quad for \: all \: |t| \leq T \: and \: all \: 0 \leq \alpha \leq 1.
   \end{equation}
  From the Duhamel expansions \eqref{equ:duhamel_expansion_vpt}, \eqref{equ:duhamel_expansion_vpta} for $\vpt$ and $\vpta$, we obtain
  \begin{equation} \label{energy_closeness_decomp}
   \begin{split}
    \| (-i\nabla + A) (\vpt - \vpta) \|_{2} \, &\leq \, |\lambda| \int_{0}^{t}{\ud s \, \left\{ \Bigl\| (-i\nabla + A) \Bigl( \frac{1}{|\cdot|} \ast |\vp_{s}|^{2} \Bigr) (\vp_{s} - \vp_{s}^{(\alpha)}) \Bigr\|_{2} \right. } \\ 
    & \quad \quad \quad \quad \quad \quad + \Bigl\| (-i\nabla + A) \Bigl( \frac{\alpha}{|\cdot| (|\cdot| + \alpha)} \ast |\vp_{s}|^{2} \Bigr) \vp_{s}^{(\alpha)} \Bigr\|_{2} \\
    & \quad \quad \quad \quad \quad \quad + \left. \Bigl\| (-i\nabla + A) \Bigl( \frac{1}{|\cdot| + \alpha} \ast (|\vp_{s}|^{2} - |\vp_{s}^{(\alpha)}|^{2}) \Bigr) \vp_{s}^{(\alpha)} \Bigr\|_{2} \right\}.
  \end{split}
  \end{equation}
  In what follows, we will be tacitly using the magnetic Hardy inequality \eqref{equ:magnetic_hardy_inequality} and the uniform $\hoa$-norm control \eqref{equ:uniform_hoa_bound}. The first term in the parenthesis on the r.h.s. of \eqref{energy_closeness_decomp} is bounded by
  \begin{equation} \label{equ:first_term}
   \begin{split}
    & \Bigl\| (-i\nabla + A) \Bigl( \frac{1}{|\cdot|} \ast |\vp_{s}|^{2} \Bigr) (\vp_{s} - \vp_{s}^{(\alpha)}) \Bigr\|_{2} \\
    \leq & \: \Bigl\| (-i \nabla) \Bigl( \frac{1}{|\cdot|} \ast |\vps|^{2} \Bigr) \Bigr\|_{\infty} \, \| \vps - \vpsa \|_{2} + \Bigl\| \frac{1}{|\cdot|} \ast |\vps|^{2} \Bigr\|_{\infty} \, \| (-i \nabla + A) (\vps - \vpsa) \|_{2} \\
    \lesssim & \: \|\vps\|_{\hoa}^2 \| \vps - \vpsa \|_{2} + \|\vps\|_{\hoa}^2 \, \| (-i \nabla + A) (\vps - \vpsa) \|_{2} \, \lesssim \: \alpha + \| (-i \nabla + A) (\vps - \vpsa) \|_{2},
   \end{split}
  \end{equation}
  where we used \eqref{equ:l2_closeness} and 
  \begin{equation*}
   \sup_{x \in \bR^{3}}{ \Bigl| \int{\ud y \, (-i \nabla_{x}) \frac{1}{|x-y|} |\vps(y)|^{2} }\Bigr| } \, \leq \, \sup_{x \in \bR^{3}}{ \int{\ud y \, \frac{1}{|x-y|^{2}} |\vps(y)|^{2} }} \, \lesssim \|\vps\|_{\hoa}^{2}.
  \end{equation*}
  
  The second term in the parenthesis on the r.h.s of \eqref{energy_closeness_decomp} is controlled by
  \begin{equation} \label{equ:second_term}
   \begin{split}
    & \Bigl\| (-i\nabla + A) \Bigl( \frac{\alpha}{|\cdot| (|\cdot| + \alpha)} \ast |\vp_{s}|^{2} \Bigr) \vp_{s}^{(\alpha)} \Bigr\|_{2} \\
    \lesssim & \: \Bigl\| (-i\nabla) \frac{\alpha}{|\cdot| (|\cdot| + \alpha)} \ast |\vps|^{2} \Bigr\|_{3} \|\vpsa\|_{6} + \alpha \, \Bigl\| \frac{1}{|\cdot|^{2}} \ast |\vps|^{2} \Bigr\|_{\infty}  \|(-i\nabla+A) \vpsa\|_2 \\ 
    \lesssim & \: \alpha \, \Bigl\| \frac{1}{|\cdot|^{2}} \ast (\nabla |\vps|^{2}) \Bigr\|_{3} \|\vpsa\|_{\hoa} + \alpha \, \|\vps\|_{\hoa}^2 \|\vpsa\|_{\hoa} \, \lesssim \, \alpha \, \| \nabla (|\vps|^{2}) \|_{3/2} + \alpha \\
    \lesssim & \: \alpha \, \|\vps\|_{6} \|\nabla |\vps| \|_{2} + \alpha \, \lesssim \, \alpha \, \|\vps\|_{\hoa}^2 + \alpha \, \lesssim \, \alpha. 
   \end{split}
  \end{equation}
  Here we used the Hardy-Littlewood-Sobolev inequality in the third estimate.

  The third term in the parenthesis on the r.h.s of \eqref{energy_closeness_decomp} is estimated as follows
  \begin{equation} \label{third_term}
   \begin{split}
    & \Bigl\| (-i\nabla + A) \Bigl( \frac{1}{|\cdot| + \alpha} \ast (|\vp_{s}|^{2} - |\vp_{s}^{(\alpha)}|^{2}) \Bigr) \vp_{s}^{(\alpha)} \Bigr\|_{2} \\
     \leq &\: \Bigl\| \Bigl( (-i\nabla) \frac{1}{|\cdot| + \alpha} \ast (|\vp_{s}|^{2} - |\vp_{s}^{(\alpha)}|^{2}) \Bigr) \vpsa \Bigr\|_{2} + \Bigl\| \frac{1}{|\cdot| + \alpha} \ast (|\vp_{s}|^{2} - |\vp_{s}^{(\alpha)}|^{2}) \Bigr\|_{\infty} \, \| (-i \nabla + A) \vpsa \|_{2}.
   \end{split}
   \end{equation}
  The first term on the r.h.s. of \eqref{third_term} is bounded by
  \begin{equation} \label{third_term_two}
   \begin{split}
    & \Bigl\| \Bigl( (-i\nabla) \frac{1}{|\cdot| + \alpha} \ast (|\vp_{s}|^{2} - |\vp_{s}^{(\alpha)}|^{2}) \Bigr) \vpsa \Bigr\|_{2}\\ 
    \leq & \: \Bigl\| (-i\nabla) \frac{1}{|\cdot| + \alpha} \ast (|\vps|^2 - |\vpsa|^2) \Bigr\|_3 \, \|\vpsa\|_6 \\
    \lesssim & \: \Bigl\| \int{\ud y \, \frac{1}{(|\cdot - y| + \alpha)^2} (|\vps(y)| + |\vpsa(y)|) |\vps(y) - \vpsa(y)|} \Bigr\|_3 \, \|\vpsa\|_{\hoa} \\
    \lesssim & \: \Bigl\| \Bigl( \int{\ud y \, \frac{1}{(|\cdot - y| + \alpha)^4} (|\vps(y)|^2 + |\vpsa(y)|^2)} \Bigr)^{1/2} \Bigr\|_3 \, \|\vps - \vpsa\|_2 \\
    \lesssim & \: \Bigl( \Bigl\| \frac{1}{|\cdot|^{5/2}} \ast |\vps|^2 \Bigr\|_{3/2}^2 + \Bigl\| \frac{1}{|\cdot|^{5/2}} \ast |\vpsa|^2 \Bigr\|_{3/2}^2\Bigr) \, \alpha^{1/4} \, \lesssim \, \Bigl( \|\vps\|_{12/5}^4 + \|\vpsa\|_{12/5}^4 \Bigr) \, \alpha^{1/4} \, \lesssim \, \alpha^{1/4}.
   \end{split}
  \end{equation}
  Here, the fourth estimate followed from \eqref{equ:l2_closeness}. In the fifth estimate, we used the Hardy-Littlewood-Sobolev inequality. The last inequality followed from $L^{p}$-interpolation between $L^{2}$ and $L^{6}$. The $L^{2}$-norm of $\vpsa$ is equal to one by mass conservation, the $L^{6}$-norm of $\vpsa$ is bounded using the Sobolev inequality and the uniform $\hoa$-norm control \eqref{equ:uniform_hoa_bound}. \\
  In order to estimate the second term on the r.h.s. of \eqref{third_term}, we observe that by \eqref{equ:l2_closeness},
  \begin{equation} \label{third_term_one}   
    \Bigl\| \frac{1}{|\cdot|+\alpha} \ast (|\vps|^2 - |\vpsa|^2) \Bigr\|_{\infty} \, \lesssim (\|\vps\|_{\hoa} + \|\vpsa\|_{\hoa}) \|\vps - \vpsa\|_2 \, \lesssim \, \alpha.
  \end{equation}
  Inserting \eqref{third_term_one} and \eqref{third_term_two} into \eqref{third_term} and using \eqref{equ:uniform_hoa_bound} yields
  \begin{equation} \label{equ:third_term}
  \Bigl\| (-i\nabla + A) \Bigl( \frac{1}{|\cdot| + \alpha} \ast (|\vp_{s}|^{2} - |\vp_{s}^{(\alpha)}|^{2}) \Bigr) \vp_{s}^{(\alpha)} \Bigr\|_{2} \, \lesssim \, \alpha^{1/4}.
  \end{equation}

  Combining \eqref{equ:first_term}, \eqref{equ:second_term} and \eqref{equ:third_term} gives
  \begin{equation*}
   \| (-i\nabla + A) (\vpt - \vpta) \|_{2} \, \lesssim \, \int_{0}^{t}{\ud s \, \Bigl\{ \| (-i\nabla + A) (\vps - \vpsa) \|_{2} + \alpha^{1/4} \Bigr\} }.
  \end{equation*}
  The claim now follows with Gronwall's lemma.
\end{proof}

\medskip

Moreover, we derive uniform estimates on the $\hoa$-norm of the time derivative of solutions to the regularized magnetic Hartree equation \eqref{equ:regularised_hartree_equ}. These are needed in the proof of Theorem \ref{thm:energy_convergence}. For $j, k, l \in \{1,2,3\}$, write $\cha \, \equiv \, (-i\nabla + A)^2 \, = \, \sum_l{D_l^2}$ and $B_{jk} \, \equiv \, (\partial_j A_k) - (\partial_k A_j)$. We will be using the commutators
\begin{align}
 [D_j, D_k] \, &= \, -i B_{jk}, \\
 [D_j, D_l^2] \, &= \, - 2 i B_{jl} D_l - (\partial_l B_{jl}),  \label{equ:DjDlsquared_commutator} \\
 [D_j D_k, D_l^2] \, &= \, -2 (\partial_j B_{kl}) D_l - 2i B_{kl} D_j D_l + i (\partial_j \partial_l B_{kl}) - (\partial_l B_{kl}) D_j - 2i B_{jl} D_l D_k - (\partial_l B_{jl}) D_k.   \label{DjDkDlsquared_commutator}
\end{align}

\begin{proposition} \label{prop:hoa_bound_time_derivative}
 Let $A \in C^{\infty}(\bR^3; \bR^3)$ satisfy assumption (A). Let $\vp \in H^3_A(\bR^3)$ with $\|\vp\|_2 = 1$. Denote by $\vpta$ the solution to the regularized magnetic Hartree equation with $\alpha \geq 0$ and initial datum $\vp_{t=0}^{(\alpha)} = \vp$. Let $T > 0$. Then there exists a constant $C \equiv C(T, \|\vp\|_{H^3_A})$ such that 
 \begin{equation} \label{equ:hoa_bound_time_derivative}
  \|\partial_t \vpta\|_{\hoa} \, \leq \, C \quad for \: all \: |t| \leq T \: and \: all \: 0 \leq \alpha \leq 1.
 \end{equation} 
\end{proposition}

The proof of Proposition \ref{prop:hoa_bound_time_derivative} relies on the following higher regularity result for the regularized magnetic Hartree equation \eqref{equ:regularised_hartree_equ}.

\begin{lemma}[Propagation of $H^2_A$-regularity] \label{lemma:h2a_regularity}
Let $A \in C^{\infty}(\bR^3; \bR^3)$ satisfy assumption (A) and let $\vp \in H^2_A(\bR^3)$ with $\|\vp\|_2 = 1$. Denote by $\vpta$ the solution to the regularized magnetic Hartree equation \eqref{equ:regularised_hartree_equ} with initial datum $\vp_{t=0}^{(\alpha)} = \vp$. Let $T > 0$. Then there exists a constant $C \equiv C(T, \|\vp\|_{H^2_A})$ such that for all $j, k \in \{1,2,3\}$,
 \begin{equation}
  \| D_j D_k \vpta \|_2 \, \leq \, C \quad for \: all \: |t| \leq T \: and \: all \: 0 \leq \alpha \leq 1.
 \end{equation}
\end{lemma}
\begin{proof}
From \eqref{equ:regularised_hartree_equ} we obtain
\begin{equation}
  \frac{\mathrm{d}}{\mathrm{d}t} \, \|D_j D_k \vpta \|_2^2 \, = \, 2 \, \im \, \Bigl\langle D_j D_k \vpta, \Bigl( [D_j D_k, \cha] \vpta + D_j D_k \Bigl( \frac{\lambda}{|\cdot|+\alpha} \ast |\vpta|^2 \Bigr) \vpta \Bigr) \Bigr\rangle.
\end{equation}
Taking the absolute value, we find
\begin{equation} \label{equ:time_derivative_phi_t_estimate}
 \Bigl| \frac{\mathrm{d}}{\mathrm{dt}}  \, \|D_j D_k \vpta \|_2^2 \Bigr| \, \leq \, 2 \|D_j D_k \vpta\|_2^2 + \|[D_j D_k, \cha]\vpta\|_2^2 + \Bigl\|D_j D_k \Bigl( \frac{\lambda}{|\cdot| + \alpha} \ast |\vpta|^2 \Bigr) \vpta \Bigr\|_2^2. 
\end{equation}
In order to bound the second term on the r.h.s of \eqref{equ:time_derivative_phi_t_estimate}, we use \eqref{DjDkDlsquared_commutator} and obtain
\begin{equation} \label{equ:DjDkcha_commutator_estimate}
 \begin{split}
  & \| [D_jD_k, \cha] \vpta \|_2 \, \leq \, \sum_l{\, \|[D_j D_k, D_l^2] \vpta\|_2} \\
  &\lesssim \, \sum_l \, \Bigl\{ \|\partial_j B_{kl}\|_{\infty} \|D_l \vpta\|_2 + \|B_{kl}\|_{\infty} \|D_j D_l \vpta\|_2 + \|\partial_j \partial_l B_{kl}\|_{\infty} \|\vpta\|_2 \\
  & \quad \quad \quad + \|\partial_l B_{kl}\|_{\infty} \|D_j \vpta\|_2 + \|B_{jl}\|_{\infty} \|D_l D_k \vpta\|_2 + \|\partial_l B_{jl}\|_{\infty} \|D_k \vpta\|_2 \Bigr\} \\
  &\lesssim \, \sum_l \, \Bigl\{ 1 + \|D_j D_l \vpta\|_2 + \|D_l D_k \vpta\|_2 \Bigr\} \, \lesssim \, \sum_{l, m} \, \Bigl\{ 1 + \|D_l D_m \vpta\|_2 \Bigr\}.
 \end{split}
\end{equation}
Here we used the boundedness of all derivatives of the vector potential $A$ by assumption (A).

To estimate the third term on the r.h.s of \eqref{equ:time_derivative_phi_t_estimate} we note that in general,
\begin{equation}
 D_j D_k (f g) \, = \, (-\partial_j \partial_k f) g + (-i \partial_k f) (D_j g) + (-i \partial_j f) (D_k g) + f (D_j D_k g). 
\end{equation}
Thus, $D_j D_k$ acts on the Hartree nonlinearity by
\begin{equation} \label{equ:DjDk_on_nonlinearity}
\begin{split} 
 & \Bigl\| D_j D_k \Bigl( \frac{\lambda}{|\cdot| + \alpha} \ast |\vpta|^2 \Bigr) \vpta \Bigr\|_2 \, \\
 &\leq  \, \Bigl\|\Bigl( - \partial_j \partial_k \frac{\lambda}{|\cdot| + \alpha} \ast |\vpta|^2 \Bigr) \vpta \Bigr\|_2 + \Bigr\|\Bigl(-i\partial_k \frac{\lambda}{|\cdot| + \alpha} \ast |\vpta|^2 \Bigr) D_j \vpta \Bigr\|_2 \\
 & \quad\quad + \Bigl\|\Bigl(-i\partial_j \frac{\lambda}{|\cdot| + \alpha} \ast |\vpta|^2 \Bigr) D_k \vpta \Bigr\|_2 + \Bigl\|\Bigl(\frac{\lambda}{|\cdot| + \alpha} \ast |\vpta|^2 \Bigr) D_j D_k \vpta \Bigr\|_2.
\end{split}
\end{equation}
Performing similar estimates as before in this section, the first three terms on the r.h.s. of \eqref{equ:DjDk_on_nonlinearity} can be controlled using the uniform $\hoa$-norm control \eqref{equ:uniform_hoa_bound}, hence
\begin{equation} \label{equ:DjDk_on_nonlinearity_estimate}
 \Bigl\| D_j D_k \Bigl( \frac{\lambda}{|\cdot| + \alpha} \ast |\vpta|^2 \Bigr) \vpta \Bigr\|_2 \, \lesssim \, 1 + \|D_j D_k \vpta\|_2.
\end{equation} 

Inserting \eqref{equ:DjDkcha_commutator_estimate} and \eqref{equ:DjDk_on_nonlinearity_estimate} into \eqref{equ:time_derivative_phi_t_estimate}, we obtain
\begin{equation*}
  \Bigl| \frac{\mathrm{d}}{\mathrm{dt}}  \, \|D_j D_k \vpta \|_2^2 \Bigr| \, \lesssim \, 1 + \sum_{l, m} \, \|D_l D_m \vpta\|_2^2
\end{equation*}
and therefore
\begin{equation} \label{equ:gronwall_estimate_h2a}
  \Bigl| \frac{\mathrm{d}}{\mathrm{d}t}  \, \sum_{j, k} \, \|D_j D_k \vpta \|_2^2 \Bigr| \, \lesssim \, 1 + \sum_{j, k} \, \|D_j D_k \vpta\|_2^2.
\end{equation}
Since the constants in \eqref{equ:gronwall_estimate_h2a} are independent of $0 \leq \alpha \leq 1$ and since $\sum_{j,k} \|D_j D_k \vp\|_2 < \infty$ by assumption, the claim follows with Gronwall's lemma.
\end{proof}

In what follows we use the shorthand notations $\phi_t \equiv \vpta$ and $\dot{\phi}_t \equiv \partial_t \phi_t$.
\begin{proof}[Proof of Proposition \ref{prop:hoa_bound_time_derivative}]
Fix $T > 0$. From \eqref{equ:regularised_hartree_equ} it follows that
\begin{equation*}
 \|\dot{\pt}\|_2 \, \lesssim \, \|\cha\pt\|_2 + \Bigl\| \frac{\lambda}{|\cdot| + \alpha} \ast |\pt|^2 \Bigr\|_{\infty} \|\pt\|_2 \, \lesssim \, \|\cha\pt\|_2 + \|\pt\|_{\hoa}^2 \|\pt\|_2 \, \lesssim \|\cha\pt\|_2 + 1.
\end{equation*}
Applying Lemma \ref{lemma:h2a_regularity}, we obtain a constant $C \equiv C(T, \|\vp\|_{H^2_A})$ such that
\begin{equation} \label{equ:l2_bound_time_derivative}
 \|\partial_t \vpta\|_2 \, \leq \, C \quad for \: all \: |t| \leq T \: and \: all \: 0 \leq \alpha \leq 1.   
\end{equation}

It remains to estimate $\|(-i\nabla + A) \partial_t \vpta \|_2$. From \eqref{equ:regularised_hartree_equ} we compute for $j \in \{1,2,3\}$,
\begin{equation}
 \begin{split}
  \frac{\ud}{\ud t} \, \| D_j \dot{\pt} \|_2^2 \, &= \, 2 \, \im \, \Bigl( \Bigl\langle D_j \dot{\pt}, [D_j, \cha] \dot{\pt} \Bigr\rangle + \Bigl\langle D_j \dot{\pt}, D_j \Bigl( \frac{\lambda}{|\cdot|+\alpha} \ast (\dot{\pt} \overline{\pt} + \pt \dot{\overline{\pt}}) \Bigr) \pt \Bigr\rangle \\
  &\quad \quad \quad \quad \quad+ \Bigl\langle D_j \dot{\pt}, D_j \Bigl( \frac{\lambda}{|\cdot|+\alpha} \ast |\pt|^2 \Bigr) \dot{\pt} \Bigr\rangle \Bigr). 
 \end{split}
\end{equation}
Taking the absolute value, we get
\begin{equation} \label{equ:Dj_time_derivative_phi_t}
 \begin{split}
  \Bigl| \frac{\ud}{\ud t} \, \|D_j \dot{\pt}\|_2^2 \Bigr| \, &\lesssim \, \|D_j \dot{\pt}\|_2^2 + \Bigl\|[D_j, \cha] \dot{\pt}\Bigr\|_2^2 + \Bigl\| D_j \Bigl( \frac{\lambda}{|\cdot|+\alpha} \ast (\dot{\pt} \overline{\pt} + \pt \dot{\overline{\pt}}) \Bigr) \pt \Bigr\|_2^2 \\
  &\quad \quad + \Bigl\| D_j \Bigl( \frac{\lambda}{|\cdot|+\alpha} \ast |\pt|^2 \Bigr) \dot{\pt} \Bigr\|_2^2.
 \end{split}
\end{equation}
The second term on the r.h.s. of \eqref{equ:Dj_time_derivative_phi_t} is bounded by
\begin{equation} \label{equ:second_term_nabla_A_bound}
 \begin{split}
  \Bigl\| [D_j, \cha] \dot{\pt} \Bigr\|_2 \, &\leq \, \sum_l \, \Bigl\{ 2 \|B_{jl}\|_{\infty} \|D_l \dot{\pt}\|_2 + \|\partial_l B_{jl}\|_{\infty} \|\dot{\pt}\|_2 \Bigr\} \\
  &\lesssim \, \sum_l \, \Bigl\{ \|D_l \dot{\pt}\|_2 + \|\dot{\pt}\|_2 \Bigr\} \lesssim \, 1 + \sum_l \, \|D_l \dot{\pt}\|_2,
 \end{split}
\end{equation}
where we used \eqref{equ:DjDlsquared_commutator}, \eqref{equ:l2_bound_time_derivative} and the boundedness of all derivatives of the magnetic vector potential by assumption (A). \\
The third term on the r.h.s. of \eqref{equ:Dj_time_derivative_phi_t} is controlled as follows
\begin{equation} \label{equ:third_term_nabla_A_bound}
 \begin{split}
  \Bigl\| D_j \Bigl( \frac{\lambda}{|\cdot|+\alpha} \ast (\dot{\pt} \overline{\pt} + \pt \dot{\overline{\pt}}) \Bigr) \pt \Bigr\|_2 \, &\lesssim \, \Bigl\| \frac{1}{|\cdot|^2} \ast (|\dot{\pt}| |\pt|) \Bigr\|_3 \|\pt\|_6 + \Bigl\|\frac{1}{|\cdot|} \ast (|\dot{\pt}| |\pt|) \Bigr\|_6 \|D_j \pt\|_3 \\
  &\lesssim \, \Bigl\| |\dot{\pt}| |\pt| \Bigr\|_{3/2} \|\pt\|_{\hoa} + \Bigl\| |\dot{\pt}| |\pt| \Bigr\|_{6/5} \|D_j \pt\|_3 \\
  &\lesssim \, \|\dot{\pt}\|_2 \|\pt\|_6 \|\pt\|_{\hoa} + \|\dot{\pt}\|_2 \|\pt\|_3 \|\pt\|_{H^2_A} \, \lesssim \, 1.
 \end{split}
\end{equation}
Here, we used the Hardy-Littlewood-Sobolev inequality in the second estimate. The $H^2_A$-regularity from Lemma \ref{lemma:h2a_regularity} and \eqref{equ:l2_bound_time_derivative} crucially entered the last estimate. \\
Using \eqref{equ:l2_bound_time_derivative}, the fourth term on the r.h.s. of \eqref{equ:Dj_time_derivative_phi_t} is estimated by
\begin{equation} \label{equ:fourth_term_nabla_A_bound}
 \begin{split}
  \Bigl\| D_j \Bigl( \frac{\lambda}{|\cdot|+\alpha} \ast |\pt|^2 \Bigr) \dot{\pt} \Bigr\|_2 \, \lesssim \, \Bigl\| \frac{1}{|\cdot|^2} \ast |\pt|^2 \Bigr\|_{\infty} \|\dot{\pt}\|_2 + \Bigl\| \frac{1}{|\cdot|} \ast |\pt|^2 \Bigr\|_{\infty} \|D_j \dot{\pt}\|_2 \, \lesssim \, 1 + \|D_j \dot{\pt}\|_2.
 \end{split}
\end{equation}

Inserting \eqref{equ:second_term_nabla_A_bound}, \eqref{equ:third_term_nabla_A_bound} and \eqref{equ:fourth_term_nabla_A_bound} into  \eqref{equ:Dj_time_derivative_phi_t} gives
\begin{equation*}
 \Bigl| \frac{\ud}{\ud t} \, \|D_j \dot{\pt}\|_2^2 \Bigr| \, \lesssim \, 1 + \sum_l \, \|D_l \dot{\pt}\|_2^2
\end{equation*}
and thus,
\begin{equation} \label{equ:nabla_A_bound_time_derivative_gronwall}
 \Bigl| \frac{\ud}{\ud t} \sum_j \, \|D_j \dot{\pt} \|_2^2 \Bigr| \, \lesssim \, 1 + \sum_j \, \|D_j \dot{\pt}\|_2^2.
\end{equation}
The assumptions about the initial datum imply $\|D_j \dot{\phi}_{t=0}\|_2 < \infty$. Since the constants in \eqref{equ:nabla_A_bound_time_derivative_gronwall} are independent of $0 \leq \alpha \leq 1$, Gronwall's lemma then yields the assertion. 
\end{proof}

\section{Convergence in trace norm} \label{sec:trace_convergence}
\setcounter{equation}{0}

\begin{proof}[Proof of Theorem \ref{thm:trace_convergence}]
We apply results from \cite{KP} to our mean-field system with an external magnetic field. For the convenience of the reader we state below the version of Theorem 3.1 and Corollary 3.2 in \cite{KP} that we will use.

\begin{theorem}[Knowles-Pickl, \cite{KP}] \label{thm:kp}
 Consider the mean-field Hamiltonian 
 \begin{equation} \label{equ:kp_hamiltonian}
  H_N \, = \, \sum_{j=1}^N{\,h_j} + \frac{1}{N} \sum_{i<j}{\, V(x_i - x_j)} \, \equiv \, H_N^0 + H_N^V
 \end{equation}
 on $L^2(\bR^{3N})$, where $h$ is a one-particle operator and $V$ is an interaction potential. We make the following assumptions.
 \begin{itemize}
 \item[(A1)] The one-particle Hamiltonian $h$ is self-adjoint and bounded from below. Without loss of generality we assume that $h \geq 0$. We define the Hilbert space $X_{N} = \mathcal{Q}(H_{N}^{0})$ as the form domain of $H_{N}^{0}$ with norm
 \begin{equation*}
  \| \psi \|_{X_{N}} \, := \, \|(\mathds{1} + H_{N}^{0})^{1/2} \psi \|_{2}.
 \end{equation*}
 \item[(A2)] The Hamiltonian \eqref{equ:kp_hamiltonian} is self-adjoint and bounded from below. We also assume that $\mathcal{Q}(H_{N}) \subset X_{N}$.
 \item[(A3)] The interaction potential $V$ is a real and even function satisfying
  \begin{equation*} \label{A3'}
   \| V^{2} \ast |\vp|^{2} \|_{\infty} \, \leq \, K \, \| \vp \|_{X_{1}}^{2}
  \end{equation*}
  for some constant $K > 0$. Without loss of generality we assume that $K \geq 1$.
 \item[(A4)] Let $\vp \in X_1$ with $\|\vp\|_2 = 1$. The solution $\vpt$ of the initial value problem for the Hartree equation
 \begin{equation*} \label{equ:kp_hartree}
  i \partial_t \vpt \, = \, h \vpt + (V \ast |\vpt|^2) \vpt
 \end{equation*}
 with initial datum $\vp_{t=0} = \vp$ satisfies
  \begin{equation*}
   \vpt \in C(\mathbb{R}; X_{1}) \cap C^{1}(\mathbb{R}; X_{1}^{*}).
  \end{equation*}
  Here, $X_1^*$ denotes the dual space of $X_1$, i.e. the closure of $L^2$ under the norm $\|\vp\|_{X_1^*} = \|(\mathds{1}+h)^{-1/2} \vp\|_2$.
 \end{itemize}
 Set $\pn = \vp^{\otimes N}$ and let $\pnt = e^{-i H_N t} \pn$. Denote by $\gntk$ the $k$-particle marginal densities associated with $\pnt$. Then
 \begin{equation*}
  \tr{\, \Bigl| \gntk - |\vpt\rangle\langle\vpt|^{\otimes k}\Bigr|} \, \leq \, \sqrt{8} \sqrt{\frac{k}{N}} e^{\phi(t)/2}
 \end{equation*}
 with $\phi(t) = 32 K \int_0^t{\ud s \, \|\vp(s)\|_{X_1}^2 }$. 
\end{theorem}

\medskip

We now verify (A1)-(A4) of Theorem \ref{thm:kp}. Note that the form domain $X_{1}$ is the magnetic Sobolev space $\hoa(\bR^3)$.
\begin{itemize}
 \item[(A1)] Under the assumption (A), the one-particle operator $h = (-i\nabla + A)^{2}$ is positive and self-adjoint by Theorem 2 in Leinfelder and Simader \cite{LS81}.
 \item[(A2)] Theorem X.16 and Example 2 in Section X.2 in \cite{RS_02} show that the operator $H_{N}^{V} = \sum_{i<j}^{N}{\frac{\lambda}{|x_{i} - x_{j}|}}$ is infinitesimally small with respect to the operator $\sum_{j=1}^{N}{-\Delta_{j}}$. Theorem 2.4 in \cite{AHS78} then implies that $H_{N}^{V}$ is also infinitesimally small with respect to $H_{N}^{0} = \sum_{j=1}^{N}{(-i\nabla_{x_j}+A(x_j))^{2}}$. Hence, by the Kato-Rellich Theorem, $H_{N}$ is self-adjoint on the domain $D(H_{N}^{0})$ of $H_{N}^{0}$ and bounded from below. Moreover, this implies that $H_{N}^{0}$ is $H_{N}$-bounded and therefore $\mathcal{Q}(H_{N}) \subset \mathcal{Q}(H_{N}^{0})$. 
 \item[(A3)] For every $\vp \in \hoa(\bR^3)$ we have
  \begin{equation*}
    \| V^{2} \ast |\vp|^2 \|_{\infty} \, = \, \sup_{x \in \bR^3}{\left| \int_{\bR^3}{\frac{\lambda^{2}}{|x-y|^{2}} |\varphi(y)|^{2} \, \ud y} \right|}  \, \leq \, 4 \lambda^2 \|\vp\|_{\hoa}^2
  \end{equation*}
  by the Hardy inequality, the translational invariance of $\nabla$ and the diamagnetic inequality \eqref{equ:diamagnetic_inequality}.
 \item[(A4)] Proposition \ref{prop:gwp_hartree} states that the solution $\vpt$ of the magnetic Hartree equation \eqref{equ:hartree_equ} with initial datum $\vp_{t=0} = \vp$ satisfies
   \begin{equation*}
    \vpt \in C(\mathbb{R}; \hoa) \cap C^{1}(\mathbb{R}; H_{A}^{-1})
   \end{equation*}
   and that furthermore, we have $\sup{ \{\| \vpt \|_{\hoa} \, | \, t \in \mathbb{R} \}} < \infty$. 
\end{itemize}
Hence, for every $k \in \mathbb{N}$ and $t \in \mathbb{R}$, we have
\begin{equation*}
 \tr{ \, \left| \gntk - |\vpt \rangle \langle \vpt|^{\otimes k} \right| } \, \leq \, \sqrt{8} \sqrt{\frac{k}{N}} e^{\phi(t)/2} \, \leq \, \sqrt{8} \sqrt{\frac{k}{N}} e^{C t}
\end{equation*}
with $C \equiv 16 K \left( \sup{ \{\| \vpt \|_{\hoa} \, | \, t \in \mathbb{R} \}} \right)^{2}$, which completes the proof.
\end{proof}

\section{Energy convergence} \label{sec:energy_convergence}
\setcounter{equation}{0}
The proof of Theorem \ref{thm:energy_convergence} is based on a Fock space representation of the many-body system. This approach to show convergence in energy first appeared in \cite{MS} relying on results in \cite{RS09}. We follow their argument and adapt it to the magnetic case.

\subsection{Fock space representation}
The bosonic Fock space over $L^2(\bR^3)$ is defined as 
\begin{equation*}
 \cF \, = \, \bigoplus_{n=0}^{\infty}{\, L^2_{s}(\bR^{3n})} \, = \, \bC \oplus \bigoplus_{n=1}^{\infty}{\, L^2_{s}(\bR^{3n})},
\end{equation*}
where $L^2_{s}(\bR^{3n})$ is the space of symmetric square-integrable functions over $\bR^{3n}$.  Elements of $\cF$ are sequences $\psi = \{\psi^{(n)}\}_{n=0}^{\infty}$ with $\psi^{(n)} \in L^2_s(\bR^{3n})$. $\cF$ is a Hilbert space with the scalar product
\begin{equation*}
 \langle \psi_1, \psi_2 \rangle_{\cF} \, = \, \sum_{n=0}^{\infty} \langle \psi_1^{(n)}, \psi_2^{(n)} \rangle_{L^2(\bR^{3n})}.
\end{equation*}
The vector $\{1, 0, \ldots\} \in \cF$ is called the vacuum and denoted by $\Omega$. 

\medskip

For arbitrary $f \in L^2(\bR^3)$ we define the creation operator $a^*(f)$ and the annihilation operator $a(f)$ on $\cF$ such that they satisfy the canonical commutation relations
\begin{equation*}
 [a(f), a^*(g)] \, = \, \langle f, g \rangle_{L^2(\bR^3)}, \quad [a(f), a(g)] \, = \, [a^*(f), a^*(g)] \, = \, 0.
\end{equation*}
We also define the operator valued distributions $a^*_x$ and $a_x$ for $x \in \bR^3$ such that the canonical commutation relations assume the form
\begin{equation*}
 [a_x, a^*_y] \, = \, \delta(x-y), \quad \quad [a_x, a_y] \, = \, [a_x^*, a_y^*] \, = \, 0.
\end{equation*}
The number of particle operator $\cN$, expressed through the distributions $a_x, a_x^*$ is given by
\begin{equation*}
 \cN \, = \, \int{\ud x \, a_x^* a_x}.
\end{equation*}

\medskip

For any sequence $\alpha = (\alpha_N)_{N \in \bN}$ with $\alpha_N \rightarrow 0$ as $N \rightarrow \infty$, we define the Hamiltonian $\chna$ on $\cF$ by $(\chna \psi)^{(n)} = H_N^{\alpha} \psi^{(n)}$, where $H_N^{\alpha}$ is the regularized Hamiltonian \eqref{equ:regularised_hamiltonian}. Using the distributions $a_x$, $a_x^*$, $\chna$ can be rewritten as
\begin{equation*}
 \chna \, = \, \int{\ud x \, a_x^* (-i\nabla_x + A(x))^2 a_x} + \frac{1}{2N} \int{\ud x \, \ud y \, \frac{\lambda}{|x-y|+\alpha_N} a_x^* a_y^* a_y a_x}.
\end{equation*}

\medskip
 
For $f \in L^2(\bR^3)$, we define the unitary Weyl-operator
\begin{equation*}
 W(f) \, = \, \exp(a^*(f) - a(f)).
\end{equation*}

\medskip

See Section 3 in \cite{MS} for a more detailed introduction to the Fock space representation of the many-body system.

\subsection{Proof of Theorem \ref{thm:energy_convergence}}
\begin{proof}[Proof of Theorem \ref{thm:energy_convergence}]
We introduce the unitary evolution $\cU_{N}(t;s)$ by the equation
\begin{equation} \label{equ:definition_unitary_evolution}
 i \partial_t \cU_N(t;s) \, = \, \cL_N(t) \cU_N(t; s) \quad \mbox{and} \quad \cU_N(s;s) \, = \, 1 \quad \mbox{for all} \, s \in \bR,
\end{equation}
with the generator
\begin{equation} \label{equ:generator_U_N}
 \begin{split}
  \cL_N(t) \, = &\, \int \ud x \, a_x^* (-i\nabla_x + A(x))^2 a_x + \int \ud x \, \Bigl( \frac{\lambda}{|\cdot| + \alpha_N} \ast |\vptan|^2\Bigr)(x) \,  a_x^* a_x \\
  &\, + \int \ud x \ud y \, \frac{\lambda}{|x-y|+\alpha_N} \, \ovptan(x) \vptan(y) a_y^* a_x \\
  &\, + \frac{1}{2} \int \ud x \ud y \, \frac{\lambda}{|x-y|+\alpha_N} \Bigl( \vptan(x) \vptan(y) a_x^* a_y^* + \ovptan(x) \ovptan(y) a_x a_y \Bigr) \\
  &\, + \frac{1}{\sqrt{N}} \int \ud x \ud y \, \frac{\lambda}{|x-y| + \alpha_N} \, a_x^* \, \Bigl( \vptan(y) a_y^* + \ovptan(y) a_y \Bigr) \, a_x \\
  &\, + \frac{1}{2N} \int \ud x \ud y \, \frac{\lambda}{|x-y|+\alpha_N} \, a_x^* a_y^* a_y a_x,
 \end{split}
\end{equation}
where $\vptan$ denotes the solution to the regularized magnetic Hartree equation \eqref{equ:regularised_hartree_equ}.

It was first observed by Hepp in \cite{H74} that
\begin{equation} \label{equ:hepp_observation}
 \cU_N^*(t;0) \, a_x \, \cU_N(t;0) \, = \, W^*(\sqrt{N}\vp) e^{i\chna t} (a_x - \sqrt{N} \vptan(x)) e^{-i\chna t} W(\sqrt{N} \vp).
\end{equation}
Using \eqref{equ:hepp_observation} it follows as in (5.8) in \cite{MS} that the integral kernel of $(-i\nabla + A) (\gnto - |\vptan\rangle \langle\vptan|) (-i\nabla+A)$ can be estimated by
\begin{equation}
 \begin{split}
  &\Bigl| \Bigl( (-i\nabla+A) (\gnto - |\vptan\rangle\langle\vptan|) (-i\nabla + A) \Bigr)(x; y)\Bigr| \\
  \leq \, & \frac{1}{\sqrt{N}} \int_0^{2\pi} \ud \theta_1 \int_0^{2 \pi} \ud \theta_2 \, \Bigl\| (-i\nabla_x + A(x)) a_x \cU_N^{\theta_1}(t; 0) \Omega \Bigr\| \, \Bigl\| (-i\nabla_y + A(y)) a_y \cU_N^{\theta_2}(t; 0) \Omega \Bigr\| \\
  &+ \frac{1}{N^{1/4}} \Bigl| (-i\nabla_y + A(y)) \vptan(y) \Bigr| \, \int_0^{2 \pi} \ud \theta \, \Bigl\| (-i\nabla_x + A(x)) a_x \cU_N^{\theta}(t;0) \Omega \Bigr\| \\
  &+ \frac{1}{N^{1/4}} \Bigl| (-i\nabla_x + A(x)) \vptan(x) \Bigr| \, \int_0^{2 \pi} \ud \theta \, \Bigl\| (-i\nabla_y + A(y)) a_y \cU_N^{\theta}(t;0) \Omega \Bigr\|. 
 \end{split}
\end{equation}
In the last equation the unitary group $\cU_N^{\theta}$ is defined by the generator \eqref{equ:generator_U_N} with $\vptan$ replaced by $e^{-i\theta} \vptan$ (note that if $\vptan$ is a solution of the regularized Hartree equation \eqref{equ:regularised_hartree_equ}, then also $e^{-i\theta} \vptan$). 

Taking the square and integrating over $x$, $y$, we find, using $\eqref{equ:uniform_hoa_bound}$,
\begin{equation}
 \begin{split}
  & \int \ud x \ud y \, \Bigl| \Bigl( (-i\nabla+A) (\gnto - |\vptan\rangle\langle\vptan|) (-i\nabla + A) \Bigr)(x; y)\Bigr|^2 \\
  & \quad \quad \quad \lesssim \, \frac{1}{N} \Bigl( \int_0^{2 \pi} \ud \theta \, \langle \cU_N^{\theta}(t; 0) \Omega, \cK \, \cU_N^{\theta}(t;0) \Omega\rangle \Bigr)^2   + \frac{1}{\sqrt{N}} \int_0^{2 \pi} \ud \theta \, \langle \cU_N^{\theta}(t;0) \Omega, \cK \, \cU_N^{\theta}(t; 0) \Omega\rangle
 \end{split}
\end{equation}
for all $|t| \leq T$. Here we defined 
\begin{equation}
 \cK \, = \, \int \ud x \, a_x^* (-i\nabla_x + A(x))^2 a_x
\end{equation}
as the kinetic energy operator. 

\medskip

The crucial Proposition \ref{prop:control_growth_kinetic_energy} below then implies that there exists $C \equiv C(T, \|\vp\|_{H^3_A})$ such that
\begin{equation}
 \Bigl\| (-i\nabla + A) (\gnto - |\vptan\rangle\langle\vptan|) (-i\nabla+A) \Bigr\|_{\mathrm{HS}} \, \leq \, \frac{C}{N^{1/4}},
\end{equation}
where $\|\cdot\|_{\mathrm{HS}}$ denotes the Hilbert-Schmidt norm. Thus, by Proposition \ref{prop:closeness_hartree_dynamics} we obtain 
\begin{equation}
 \Bigl\| (-i\nabla + A) (\gnto - |\vpt\rangle\langle\vpt|) (-i\nabla+A) \Bigr\|_{\mathrm{HS}} \, \leq \, C \, \Bigl( \frac{1}{N^{1/4}} + \alpha_N^{1/4} \Bigr).
\end{equation}
It now follows as in (5.12) -- (5.16) in \cite{MS} that
\begin{equation}
 \tr \, \Bigl| (-i\nabla + A) (\gnto - |\vpt\rangle\langle\vpt|) (-i\nabla + A) \Bigr| \, \leq \, C \, \Bigl( \frac{1}{N^{1/4}} + \alpha_N^{1/4} \Bigr),
\end{equation}
which proves the theorem.
\end{proof}

\subsection{Control of the growth of the kinetic energy}
\begin{proposition} \label{prop:control_growth_kinetic_energy}
Suppose that the assumptions of Theorem \ref{thm:energy_convergence} are satisfied. Let the unitary evolution $\cU_N(t; s)$ be defined as in \eqref{equ:definition_unitary_evolution}. Then there exists $C \equiv C(T, \|\vp\|_{H^3_A})$ such that
\begin{equation}
 \langle \cU_N(t; 0) \Omega, \cK \, \cU_N(t; 0) \Omega \rangle \, \leq \, C
\end{equation}
for all $t \in \bR$ with $|t| \leq T$.
\end{proposition}

In what follows we use the shorthand notation $\pt \equiv \vptan$.

\begin{proof}
We compare the growth of the expectation of the kinetic energy operator $\cK$ along the dynamics $\cU_N$ and along a new dynamics $\cW_N$ defined through the equation
\begin{equation} \label{equ:definition_W_N}
 i \partial_t \cW_N(t; s) \, = \, \cM_N(t) \cW_N(t; s) \quad \quad \mbox{and} \quad \cW_N(s; s) = 1 \quad \quad \mbox{for all} \, s \in \bR,
\end{equation}
with the generator
\begin{equation} \label{equ:generator_W_N}
 \begin{split}
  \cM_N(t) \, = &\, \int \ud x \, a_x^* (-i\nabla_x + A(x))^2 a_x + \int \ud x \, \Bigl( \frac{\lambda}{|\cdot| + \alpha_N} \ast |\pt|^2\Bigr)(x) \,  a_x^* a_x \\
  &\, + \int \ud x \ud y \, \frac{\lambda}{|x-y|+\alpha_N} \, \opt(x) \pt(y) a_y^* a_x \\
  &\, + \frac{1}{2} \int \ud x \ud y \, \frac{\lambda}{|x-y|+\alpha_N} \Bigl( \pt(x) \pt(y) a_x^* a_y^* + \opt(x) \opt(y) a_x a_y \Bigr) \\
  &\, + \frac{1}{\sqrt{N}} \int \ud x \ud y \, \frac{\lambda}{|x-y| + \alpha_N} \, \Bigl( \pt(y) \, a_x^* a_y^* \mathbf{1}_{\vartheta N}(\cN) a_x + \opt(y) \, a_x^* \mathbf{1}_{\vartheta N}(\cN) a_x a_y \Bigr) \\
  &\, + \frac{1}{2N} \int \ud x \ud y \, \frac{\lambda}{|x-y|+\alpha_N} \, a_x^* a_y^* \mathbf{1}_{\vartheta N}(\cN) a_y a_x,
 \end{split}
\end{equation}
where $\mathbf{1}_{\vartheta N}(\cN)$ denotes the characteristic function of the interval $(-\infty, \vartheta N]$ for some $\vartheta > 0$ to be fixed later. 

\medskip

Next, we expand
\begin{equation*}
 \begin{split}
  \langle \cU_N(t; 0) \Omega, \cK \, \cU_N(t; 0) \Omega \rangle \, &= \, \langle \cW_N(t; 0) \Omega, \cK \, \cW_N(t; 0) \Omega\rangle + \langle (\cU_N(t;0) - \cW_N(t;0)) \Omega, \cK \, \cW_N(t; 0) \Omega\rangle \\
  &\quad + \langle \cU_N(t; 0) \Omega, \cK \, (\cU_N(t; 0) - \cW_N(t; 0)) \Omega \rangle \\
  &\leq \langle \cW_N(t;0) \Omega, \cK \cW_N(t;0) \Omega \rangle + \|(\cU_N(t;0) - \cW_N(t;0))\Omega\| \, \|\cK \, \cW_N(t;0) \Omega \| \\
  &\quad + \|\cK \, \cU_N(t;0) \Omega\| \, \|(\cU_N(t;0) - \cW_N(t;0)) \Omega \|.
 \end{split}
\end{equation*}
The claim now follows from Proposition \ref{prop:growth_K2_cutoffed_dynamics}, Proposition \ref{prop:weak_bound_growth_K2}, Proposition \ref{prop:compare_dynamics} below and from the assumption that $N^{\beta} \alpha_N \rightarrow \infty$ as $N \rightarrow \infty$ for some $\beta > 0$.
\end{proof}

\begin{proposition} \label{prop:growth_K2_cutoffed_dynamics}
Suppose that the assumptions of Proposition \ref{prop:control_growth_kinetic_energy} are satisfied (but here the assumption $N^{\beta} \alpha_N \rightarrow \infty$ as $N \rightarrow \infty$ for some $\beta > 0$ will not be used). Let the evolution $\cW_N(t;s)$ be defined by \eqref{equ:definition_W_N} and suppose that $\vartheta > 0$ is small enough. Then there exists $C \equiv C(\vartheta, T, \|\vp\|_{H^3_A})$ such that
\begin{equation}
 \langle \cW_N(t; 0) \Omega, \cK^2 \cW_N(t; 0) \Omega \rangle \, \leq \, C
\end{equation}
for all $t \in \bR$ with $|t| \leq T$.
\end{proposition}

\begin{proposition} \label{prop:weak_bound_growth_K2}
Suppose that the assumptions of Proposition \ref{prop:control_growth_kinetic_energy} are satisfied (but here the assumption $N^{\beta} \alpha_N \rightarrow \infty$ as $N \rightarrow \infty$ for some $\beta > 0$ will not be used). Then there exists $C \equiv C(T, \|\vp\|_{H^2_A})$ such that 
\begin{equation}
 \langle \cU_N(t; 0) \Omega, \cK^2 \cU_N(t; 0) \Omega \rangle \, \leq \, C \, \Bigl( N^2 + \frac{N^2}{\alpha_N^2} \Bigr)
\end{equation}
for all $t \in \bR$ with $|t| \leq T$. 
\end{proposition}

\begin{proposition} \label{prop:compare_dynamics}
Suppose that the assumptions of Proposition \ref{prop:control_growth_kinetic_energy} are satisfied (but here the assumption $N^{\beta} \alpha_N \rightarrow \infty$ as $N \rightarrow \infty$ for some $\beta > 0$ will not be used). Let the evolution $\cW_N(t;s)$ be defined by \eqref{equ:definition_W_N} with $\vartheta > 0$. Then, for any $k \in \bN$, there exists $C \equiv C(k, \vartheta, T, \|\vp\|_{H^1_A})$ such that
\begin{equation}
 \Bigl\| \Bigl( \cU_N(t; 0) - \cW_N(t;0) \Bigr) \Omega \Bigr\| \, \leq \, \frac{C}{N^k} \, \Bigl( 1 + \frac{1}{\alpha_N} \Bigr)
\end{equation}
for all $t \in \bR$ with $|t| \leq T$. 
\end{proposition}

\subsubsection{Growth of $\cK^2$ with respect to $\cW_N$ dynamics}
\begin{proof}[Proof of Proposition \ref{prop:growth_K2_cutoffed_dynamics}]
In what follows, we will be repeatedly using the bounds \eqref{equ:uniform_hoa_bound} and \eqref{equ:hoa_bound_time_derivative}. Recall also the shorthand notations $\pt \equiv \vptan$ and $\dot{\pt} \equiv \partial_t \pt$.

\medskip

From the definition \eqref{equ:generator_W_N} we obtain
\begin{equation}
 \begin{split}
  \cK^2 \, \lesssim \, &\cM_N^2(t) + \Bigl( \int \ud x \, \Bigl( \frac{1}{|\cdot| + \alpha_N} \ast |\pt|^2 \Bigr)(x) \, a_x^* a_x \Bigr)^2 \\
  &+ \Bigl( \int \ud x \ud y \, \frac{1}{|x-y| + \alpha_N} \, \opt(x) \pt(y) \, a_y^* a_x \Bigr)^2 \\
  &+ \Bigl( \frac{1}{2} \int \ud x \ud y \, \frac{1}{|x-y| + \alpha_N} \, (\pt(x) \pt(y) \, a_x^* a_y^* + h.c.) \Bigr)^2 \\
  &+ \Bigl( \frac{1}{\sqrt{N}} \int \ud x \ud y \, \frac{1}{|x-y|+\alpha_N} \, (\pt(y) \, a_x^* a_y^* \mathbf{1}_{\vartheta N}(\cN) a_x + h.c.) \Bigr)^2 \\
  &+ \Bigl( \frac{1}{2N} \int \ud x \ud y \, \frac{1}{|x-y|+\alpha_N} \, a_x^* a_y^* \mathbf{1}_{\vartheta N}(\cN) a_y a_x \Bigr)^2,
 \end{split}
\end{equation}
where $h.c.$ denotes the hermitian conjugate. \\
By Lemma \ref{lemma:estimate_quartic_term} below the last term is bounded by
\begin{equation}
 \Bigl( \frac{1}{2N} \int \ud x \ud y \frac{1}{|x-y| + \alpha_N} \, a_x^* a_y^* \mathbf{1}_{\vartheta N}(\cN) a_y a_x \Bigr)^2 \, \lesssim \, \vartheta^2 (\cN^2 + \cK^2).
\end{equation}
Thus, choosing $\vartheta > 0$ sufficiently small, we get
\begin{equation} \label{equ:K2_estimate}
 \begin{split}
  \cK^2 \, \lesssim \, &\cM_N^2(t) + \cN^2 + \Bigl( \int \ud x \, \Bigl( \frac{1}{|\cdot|+\alpha_N} \ast |\pt|^2 \Bigr)(x) \, a_x^* a_x \Bigr)^2 \\
  &+ \Bigl( \int \ud x \ud y \, \frac{1}{|x-y| + \alpha_N} \, \opt(x) \pt(y) \, a_y^* a_x \Bigr)^2 \\
  &+ \Bigl( \frac{1}{2} \int \ud x \ud y \, \frac{1}{|x-y| + \alpha_N} \, (\pt(x) \pt(y) \, a_x^* a_y^* + h.c.) \Bigr)^2 \\
  &+ \Bigl( \frac{1}{\sqrt{N}} \int \ud x \ud y \, \frac{1}{|x-y|+\alpha_N} \, (\pt(y) \, a_x^* a_y^* \mathbf{1}_{\vartheta N}(\cN) a_x + h.c.) \Bigr)^2.
 \end{split}
\end{equation}
In order to bound the third term on the r.h.s. of the last equation, we observe that
\begin{equation} 
 \int \ud x \Bigl( \frac{1}{|\cdot| + \alpha_N} \ast |\pt|^2 \Bigr)(x) \, a_x^* a_x \, \leq \, \Bigl\| \frac{1}{|\cdot|+\alpha_N} \ast |\pt|^2 \Bigr\|_{\infty} \, \cN \, \lesssim \, \|\pt\|_{\hoa}^2 \, \cN.
\end{equation}
Since the number of particles operator $\cN$ commutes with the operator on the l.h.s., we conclude
\begin{equation} \label{equ:K2_estimate_third_term}
 \Bigl( \int \ud x \Bigl( \frac{1}{|\cdot| + \alpha_N} \ast |\pt|^2 \Bigr)(x) \, a_x^* a_x \Bigr)^2 \, \lesssim \, \|\pt\|_{\hoa}^4 \, \cN^2.
\end{equation}
Analogously, the fourth term on the r.h.s. of \eqref{equ:K2_estimate} is bounded by
\begin{equation} \label{equ:K2_estimate_fourth_term}
 \Bigl( \int \ud x \ud y \, \frac{1}{|x-y| + \alpha_N} \, \opt(x) \pt(y) \, a_y^* a_x \Bigr)^2 \, \lesssim \, \|\pt\|_{\hoa}^4 \, \cN^2.
\end{equation}
The terms in the third and fourth line of \eqref{equ:K2_estimate} can be estimated as in (6.13) -- (6.20) in \cite{MS}. The only difference is that here we use the bound $\bigl\|\frac{1}{|\cdot|^2} \ast |\pt|^2 \bigr\|_{\infty} \lesssim \|\pt\|_{\hoa}^2$. This yields
\begin{equation} \label{equ:K2_estimate_fifth_term}
 \Bigl( \frac{1}{2} \int \ud x \ud y \, \frac{1}{|x-y| + \alpha_N} \, (\pt(x) \pt(y) \, a_x^* a_y^* + h.c.) \Bigr)^2 \, \lesssim \, \|\pt\|_{\hoa}^4 \, (\cN + 1)^2
\end{equation}
and
\begin{equation} \label{equ:K2_estimate_sixth_term}
 \Bigl( \frac{1}{\sqrt{N}} \int \ud x \ud y \, \frac{1}{|x-y|+\alpha_N} \, (\pt(y) \, a_x^* a_y^* \mathbf{1}_{\vartheta N}(\cN) a_x + h.c.) \Bigr)^2 \, \lesssim \, \|\pt\|_{\hoa}^2 \, (\cN + 1)^3.
\end{equation}
Combining \eqref{equ:K2_estimate_third_term} -- \eqref{equ:K2_estimate_sixth_term} and using the uniform $\hoa$-norm control \eqref{equ:uniform_hoa_bound}, we obtain 
\begin{equation} \label{equ:K2_estimate_simplified}
 \cK^2 \, \lesssim \, \cM_N^2(t) + (\cN +1)^3.
\end{equation}

\medskip

Moreover, there exists a constant $C \equiv C(T, \|\vp\|_{\hoa})$ such that
\begin{equation} \label{equ:growth_N+1_cubed_wrt_W_N}
 \langle \cW_N(t;0) \Omega, (\cN+1)^3 \cW_N(t;0) \Omega \rangle \, \leq \, C
\end{equation}
for all $|t| \leq T$. The proof of this bound is analogous to the proof of Lemma 3.5 in \cite{RS09} with $M = \vartheta N$. The difference is that here we control the arising terms involving the Hartree nonlinearity $\bigl\| \frac{1}{(|\cdot| + \alpha_N)^2} \ast |\pt|^2 \bigr\|_{\infty} \lesssim \|\pt\|_{\hoa}^2$ by the uniform $\hoa$-norm bound \eqref{equ:uniform_hoa_bound}. Moreover, the generator $\cM_N(t)$ also contains a cutoff in the quartic term, but this is not relevant, since the quartic term commutes with the number of particles operator. Note also that the magnetic kinetic energy operator $\cK$ commutes with the number of particles operator.

\medskip

It remains to control the growth of the expectation of $\cM_N^2(t)$. Using \eqref{equ:definition_W_N} we compute
\begin{equation} \label{equ:M_N_squared_gronwall}
 \Bigl| \frac{\ud}{\ud t} \langle \cW_N(t;0) \Omega, \cM_N^2(t) \cW_N(t;0) \Omega \rangle^{1/2} \Bigr| \, \leq \, \langle \cW_N(t;0) \Omega, \dot{\cM}_N^2(t) \cW_N(t;0) \Omega \rangle^{1/2}.
\end{equation}
We have
\begin{equation} \label{equ:M_N_time_derivative}
 \begin{split}
  \dot{\cM}_N(t) \, = \,& \int \ud x \, \Bigl( \frac{\lambda}{|\cdot|+\alpha_N} \ast (\dot{\overline{\phi}}_t \pt + \opt \dot{\phi}_t )  \Bigr)(x) \, a_x^* a_x \\
  &+ \int \ud x \ud y \, \frac{\lambda}{|x-y|+\alpha_N} \, ( \dot{\overline{\phi}}_t(x) \pt(y) + \opt(x) \dot{\phi}_t(y)) \, a_y^* a_x \\
  &+ \int \ud x \ud y \, \frac{\lambda}{|x-y| + \alpha_N} \, (\dot{\phi}_t(x) \pt(y) \, a_x^* a_y^* + h.c.) \\
  &+ \frac{1}{\sqrt{N}} \int \ud x \ud y \, \frac{\lambda}{|x-y|+\alpha_N} \, ( \dot{\phi}_t(y) \,  a_x^* a_y^* \mathbf{1}_{\vartheta N}(\cN) a_x + h.c. ).
 \end{split}
\end{equation}
Next, we estimate the squares of the terms on the r.h.s. of \eqref{equ:M_N_time_derivative}. Similarly to \eqref{equ:K2_estimate_third_term} -- \eqref{equ:K2_estimate_sixth_term} these are all bounded by $(\cN + 1)^3$ with prefactors that are now powers of $\|\pt\|_{\hoa}$ and $\|\dot{\pt}\|_{\hoa}$. Using \eqref{equ:growth_N+1_cubed_wrt_W_N} and the crucial uniform bounds \eqref{equ:uniform_hoa_bound} and \eqref{equ:hoa_bound_time_derivative} on the $\hoa$-norms of $\pt$ and $\dot{\pt}$, we obtain
\begin{equation}
 \langle \cW_N(t;0) \Omega, \dot{\cM}_N^2(t) \cW_N(t;0) \Omega \rangle \, \lesssim \, 1.
\end{equation}
Gronwall's lemma applied to \eqref{equ:M_N_squared_gronwall} then yields a constant $C \equiv C(T, \|\vp\|_{H_A^3})$ such that
\begin{equation*}
 \langle \cW_N(t;0) \Omega, \cM_N^2(t) \cW_N(t;0) \Omega \rangle \, \leq \, C
\end{equation*}
for all $|t| \leq T$.

Together with \eqref{equ:K2_estimate_simplified} and \eqref{equ:growth_N+1_cubed_wrt_W_N} the proposition follows.
\end{proof}

\begin{lemma} \label{lemma:estimate_quartic_term}
 There exists $C > 0$ such that
 \begin{equation} \label{equ:estimate_reabsorb_K}
   \Bigl( \frac{1}{2N} \int{\ud x \ud y \, \frac{1}{|x-y| + \alpha} \, a_x^* a_y^* \mathbf{1}_{\vartheta N}(\cN) a_y a_x \Bigr)^{2}} \, \leq \, C \, \vartheta^{2} \, ( \cN^{2} + \cK^{2} )
 \end{equation}
 for all $\alpha, \vartheta > 0$.
\end{lemma}
\begin{proof}
 Denote $h_{i} = (-i \nabla_{x_i} + A(x_{i}))^{2}$ and 
 \begin{equation*}
  \tilde{\cal{V}} \, = \, \frac{1}{2N} \int{\ud x \ud y \, \frac{1}{|x-y| + \alpha} \, a_x^* a_y^* \mathbf{1}_{\vartheta N}(\cN) a_y a_x}. 
 \end{equation*} 
 Then $\tilde{\cal{V}}$ (and thus $\tilde{\cal{V}}^{2}$) leaves the number of particles invariant and on the $n$-particle sector, we have
 \begin{equation*}
  \Bigl( \tilde{\cal{V}}^{2} \Bigr)^{(n)} \, = \, \Bigl( \frac{1}{N} \sum_{1 \leq i < j \leq n}{ \frac{1}{|x_i - x_j| + \alpha} } \Bigr)^{2} \quad \mbox{if } n \leq \vartheta N
 \end{equation*}
 and $( \tilde{\cal{V}}^{2} )^{(n)} = 0$, if $n > \vartheta N$. 

 %Using Lemma \ref{lemma:operator_inequality} below, we obtain
 Using the magnetic Hardy inequality
 \begin{equation*}
  \frac{1}{|x-y|^2} \, \lesssim \, 1 + (-i\nabla_x + A(x))^2
 \end{equation*}
 we obtain
 \begin{equation*}
  \begin{split}
   \Bigl( \tilde{\cal{V}}^{2} \Bigr)^{(n)} \, & \lesssim \, \frac{n^2}{N^2} \, \sum_{1 \leq i < j \leq n}{ \frac{1}{(|x_{i} - x_{j}| + \alpha)^{2}} } \lesssim \, \frac{n^2}{N^2} \, \sum_{1 \leq i < j \leq n}{ (1 + h_{i}) } \, \lesssim \, \vartheta^{2} \, \Bigl( \sum_{j=1}^{n}{\, (1 + h_{j}) } \Bigr)^{2} \\
   &= \, \vartheta^{2} \, \Bigl( n + \sum_{j=1}^{n}{\, h_{j} } \Bigr)^{2} \, = \, \vartheta^{2} \, \Bigl( (\cN + \cK)^{2} \Bigr)^{(n)} \, \lesssim \, \vartheta^{2} \, \Bigl( \cN^{2} + \cK^{2} \Bigr)^{(n)}.
  \end{split}
 \end{equation*}
\end{proof}

\subsubsection{Weak bounds on the growth of $\cK^2$ with respect to $\cU_N$ dynamics}
\begin{proof}[Proof of Proposition \ref{prop:weak_bound_growth_K2}]
Recall the shorthand notation $\pt \equiv \vptan$. Similarly to (6.34) -- (6.36) in \cite{MS} it follows for all $|t| \leq T$ that
\begin{equation} \label{equ:K2_expectation_U_N_estimate}
 \begin{split}
  \langle \cU_N(t;0) \Omega, \cK^2 \, \cU_N(t;0) \Omega \rangle \, &\lesssim \, \langle e^{-i\chna t} W(\sqrt{N} \vp) \Omega, \cK^2 \, e^{-i\chna t} W(\sqrt{N} \vp) \Omega \rangle \\
  &\quad + N \|\pt\|_{\hoa}^2 \langle W(\sqrt{N} \vp) \Omega, (\cN + 1) \, W(\sqrt{N} \vp) \Omega \rangle + N^2 \, \|\pt\|_{\hoa}^4.
 \end{split}
\end{equation}
We have
\begin{equation*}
 \chna \, = \, \cK + \mathcal{V}, \quad \mbox{where} \quad \mathcal{V} \, = \, \frac{1}{2N} \, \int \ud x \ud y \, \frac{\lambda}{|x-y| + \alpha_N} \, a_x^* a_y^* a_x a_y.
\end{equation*}
Since
\begin{equation} \label{equ:simple_V_estimate}
 \mathcal{V} \, \lesssim \, \frac{1}{N \alpha_N} \, \cN^2
\end{equation}
and since $[\mathcal{V}, \cN] = 0$, we find
\begin{equation} \label{equ:K2_operator_simple_estimate}
 \cK^2 \, \lesssim \, (\chna)^2 + \mathcal{V}^2 \, \lesssim \, (\chna)^2 + \frac{1}{N^2 \alpha_N^2} \cN^4.
\end{equation}
It is at this point that we use the regularization of the Coulomb potential. It allows us to estimate the interaction part $\mathcal{V}$ as in \eqref{equ:simple_V_estimate} and in this way to obtain the weak bound \eqref{equ:K2_operator_simple_estimate} on $\cK^2$. \\
Inserting \eqref{equ:K2_operator_simple_estimate} into \eqref{equ:K2_expectation_U_N_estimate} and using the bound \eqref{equ:uniform_hoa_bound}, we obtain for all $|t| \leq T$ that 
\begin{equation} \label{equ:K2_expectation_U_N_estimate_intermediate}
 \begin{split}
  \langle \cU_N(t;0) \Omega, \cK^2 \, \cU_N(t;0) \Omega \rangle \, \lesssim& \, \langle W(\sqrt{N} \vp) \Omega, \cK^2 \, W(\sqrt{N} \vp) \Omega \rangle + \frac{1}{N^2 \alpha_N^2} \langle W(\sqrt{N} \vp) \Omega, \cN^4 W(\sqrt{N} \vp) \Omega \rangle \\
  &\, + N \, \langle W(\sqrt{N} \vp) \Omega, (\cN + 1) W(\sqrt{N} \vp) \Omega \rangle + N^2.
 \end{split}
\end{equation}
From the properties of the Weyl operator (see e.g. Section 3 in \cite{MS}) we infer 
\begin{equation} \label{equ:weyl_operator_number_estimate}
 \langle W(\sqrt{N} \vp) \Omega, (\cN+1) W(\sqrt{N} \vp) \Omega \rangle \, \lesssim \, N \quad \mbox{and} \quad \langle W(\sqrt{N} \vp) \Omega, \cN^4 W(\sqrt{N} \vp) \Omega \rangle \, \lesssim \, N^4.
\end{equation}
Furthermore, we conclude as in (6.40) in \cite{MS} that
\begin{equation} \label{equ:K2_weyl_id}
  \langle W(\sqrt{N} \vp) \Omega, \cK^2 \, W(\sqrt{N} \vp) \Omega \rangle \, = \, N^2 \, \|(-i\nabla+A) \vp\|_2^4 + N \, \|(-i\nabla+A)^2 \vp\|_2^2.
\end{equation}

Inserting \eqref{equ:weyl_operator_number_estimate} and \eqref{equ:K2_weyl_id} into \eqref{equ:K2_expectation_U_N_estimate_intermediate} and using the assumption about the initial datum, we obtain
\begin{equation*}
 \langle \cU_N(t;0) \Omega, \cK^2 \, \cU_N(t;0) \Omega \rangle \, \lesssim N^2 + \frac{N^2}{\alpha_N^2}
\end{equation*}
for all $|t| \leq T$, which completes the proof.
\end{proof}

\subsubsection{Comparison of $\cU_N$ and $\cW_N$ dynamics}
\begin{proof}[Proof of Proposition \ref{prop:compare_dynamics}]
The proof of Proposition \ref{prop:compare_dynamics} proceeds as in (6.41) -- (6.47) in \cite{MS}. It relies on the existence of a constant $C \equiv C(k, T, \|\vp\|_{\hoa})$ such that
\begin{equation*}
 \langle \cW_N(t;0) \Omega, (\cN +1)^{4+k} \, \cW_N(t;0) \Omega \rangle \, \leq \, C
\end{equation*}
for all $|t| \leq T$, which follows similarly to Lemma 3.5 in \cite{RS09}.
\end{proof}

\end{document}